\documentclass[a4paper,numberwithinsect]{lipics-v2021}

\usepackage{amsmath} 
\usepackage{amssymb}
\usepackage{tikz}

\usepackage[noend]{algpseudocode}
\usepackage{algorithm}
\algnewcommand\algorithmicforeach{\textbf{for each}}

\usepackage{hyperref}

\usepackage{amssymb}

\usepackage{microtype}
\usepackage{todonotes}

\def\mx#1{\mbox{\boldmath$#1$}}
\let\restriction\upharpoonright \def\parc#1{\restriction\!{#1}}

\title{Twin-Width is Linear in the Poset Width}

\author{Jakub Balab\'an}{Masaryk University, Brno, Czech Republic}
		{jakbal@mail.muni.cz}{}{}
\author{Petr Hlin\v en\'y}{Masaryk University, Brno, Czech Republic}
		{hlineny@fi.muni.cz}{https://orcid.org/0000-0003-2125-1514}{}
\authorrunning{J.~Balab\'an and P.~Hlin\v en\'y}
\Copyright{Jakub Balab\'an and Petr Hlin\v en\'y}



\category{}
\supplement{}

\funding{Supported by the {Czech Science Foundation}, project no.~{20-04567S}.}

\EventAcronym{CVIT}
\EventYear{2016}
\EventDate{December 24--27, 2016}
\EventLocation{Little Whinging, United Kingdom}
\EventLogo{}
\SeriesVolume{42}
\ArticleNo{23}

 \nolinenumbers 
 \hideLIPIcs  

\def\ca#1{{\cal#1}}
\def\cO{\ca O}

\begin{document}
	
	\maketitle
	
\begin{abstract}
Twin-width is a new parameter informally measuring how diverse are the
neighbourhoods of the graph vertices, and it extends also to other binary
relational structures, e.g\. to digraphs and posets.
It was introduced just very recently, in 2020 by Bonnet, Kim, Thomass{\'{e}} and Watrigant.
One of the core results of these authors is that FO model checking on graph
classes of bounded twin-width is in FPT.
With that result, they also claimed that posets of bounded width have
bounded twin-width, thus capturing prior result on FO model checking of
posets of bounded width in FPT.
However, their translation from poset width to twin-width was indirect and
giving only a very loose double-exponential bound.
We prove that posets of width $d$ have twin-width at most $9d$ with a direct
and elegant argument, and show that this bound is asymptotically tight.
Specially, for posets of width $2$ we prove that in the worst case their
twin-width is also equal~$2$.
These two theoretical results are complemented with straightforward
algorithms to construct the respective contraction sequence for a given poset.
\end{abstract}

\section{Introduction}\label{introduction}

The new notion of twin-width (of graphs, digraphs, or matrices)
was introduced just very recently, in 2020, by Bonnet, Kim, Thomass{\'{e}}
and Watrigant~\cite{DBLP:conf/focs/Bonnet0TW20},
and yet has already found many very interesting applications.
These applications span from efficient parameterized algorithms and
algorithmic metatheorems, through finite model theory, to classical
combinato\-rial questions.
See also the series of follow-up papers
\cite{DBLP:conf/soda/BonnetGKTW21,DBLP:journals/corr/abs-2007-14161,DBLP:journals/corr/abs-2102-03117,DBLP:journals/corr/abs-2102-06880}.

We leave formal definitions for the next section.
Informally, in simple graphs,
twin-width measures how diverse are the neighbourhoods of the graph vertices.
E.g., cographs%
\footnote{Cographs are the graphs which can be built from singleton vertices
by repeated operations of a disjoint union and taking the complement.}
have the lowest possible value of twin-width,~$0$,
which means that the graph can be brought down to a single vertex by
successively identifying twin vertices (hence the name, {\em twin-width}).
Two vertices $x$ and $y$ are {\em twins} if they have the same neighbours in the
graph, precisely $N(x)\setminus\{y\}=N(y)\setminus\{x\}$
(the concept of twin-width of graphs does not care about mutual adjacency 
of the identified vertices).

More generally, 
imagine we identify arbitrary two vertices $x_1,x_2$ in a graph $G$ into a new
vertex $x$; then the ordinary neighbours of new $x$ will capture what the
former neighbourhoods of $x_1$ and of $x_2$ in $G$ have had in common
(except each other of~$x_1,x_2$), and 
we will additionally create new {\em red edges} from $x$ to those vertices on
which the former neighbourhoods of $x_1$ and of $x_2$ disagreed in~$G$.
Moreover, all other previously created red edges at $x_1$ or $x_2$ will stay 
red and incident to~$x$ after the identification.
Note that the former vertices $x_1,x_2$ are removed and no loop is created on~$x$.
Precisely, denote by $N(x_i)$ the ordinary (``black'') neighbours of $x_i$ in $G$
and by $N_r(x_i)$ the red neighbours of $x_i$.
After the identification of $x_1$ and $x_2$ into~$x$,
the ordinary neighbours of $x$ will be the vertices of
$\big(N(x_1)\cap N(x_2)\big)\setminus\big(\{x_1,x_2\}\cup N_r(x_1)\cup N_r(x_2)\big)$, 
and the red edges of $x$ will go to the vertices of 
$\big((N(x_1)\Delta N(x_2))\cup N_r(x_1)\cup N_r(x_2)\big)\setminus\{x_1,x_2\}$.
With respect to this, a graph $G$ has twin-width $\leq d$ if one can
reduce $G$ down into a single vertex by successively identifying pairs of 
its vertices such that the maximum degree of the subgraph on the red edges
at every step of this reduction process is at most~$d$.

We are, in particular, interested in the algorithmic metatheorem area.
Namely, Bonnet et al.~\cite{DBLP:conf/focs/Bonnet0TW20} have proved that classes of 
binary relational structures (such as simple graphs and digraphs) of bounded twin-width
have efficient FO model checking algorithms.
In one of previous studies on algorithmic metatheorems for dense structures,
Gajarsk\'y et al.~\cite{DBLP:conf/focs/GajarskyHLOORS15} proved that posets
(which present a special case of simple digraphs)
of bounded width admit efficient FO model checking algorithms.
The {\em width of a poset} is the maximum size of an antichain in it.
Since \cite{DBLP:conf/focs/Bonnet0TW20} have also proved that posets of bounded
width have bounded twin-width, this directly generalizes the algorithmic metatheorem
of \cite{DBLP:conf/focs/GajarskyHLOORS15}.

The proof of bounded width of posets in~\cite{DBLP:conf/focs/Bonnet0TW20} is, however, indirect
(it uses a characterization by so-called mixed minors of the adjacency matrix) and gives, 
for posets of width $d$, only a very loose upper bound of $2^{2^{\cO(d)}}$ for the twin-width.
Although the proof in~\cite{DBLP:conf/focs/Bonnet0TW20} is, in principle, constructive,
its intricacy makes it really hard to understand why posets of bounded width
should have bounded twin-width, and which vertices to identify in the
reduction process.
In fact, as we will see in this paper, already for posets of width $2$ there
is no immediate way to optimally choose the pairs for identification.

The main contribution of our paper is in giving direct and tighter constructive linear lower and 
quadratic upper bounds for the twin-width of posets of width~$d$.
Precisely, the twin-width of such a poset in the worst case is at least $d-1$ 
and at most $9d-6$ (Proposition~\ref{pro:lowerd} and Theorem~\ref{thm:maind}).
Specially for posets of width $2$, we prove that their twin-width is also
at most $2$ and this bound cannot be further improved (Theorem~\ref{thm:main2}).
These results are accompanied by simple and fast algorithms to compute
the corresponding contraction sequence.

Our refined results on twin-width of posets can bring faster theoretical
algorithms for FO model checking of posets
\cite{DBLP:conf/focs/GajarskyHLOORS15} and of other classes which have
previously been reduced to posets of bounded width, such as
\cite{DBLP:journals/corr/abs-1302-6043,DBLP:journals/tocl/BovaGS16,DBLP:journals/comgeo/HlinenyPR19}.

\section{Preliminaries and formal definitions}

We consider only finite graphs.
Our graphs and digraphs are simple, meaning that they do not have parallel
edges or loops, except that a simple digraph may have up to one oriented
loop per vertex.
Formally, a graph is a pair $G=(V,E)$ such that $E\subseteq{V\choose2}$,
and a digraph is $G=(V,E)$ such that $E\subseteq V\times V$.
We deal with (finite) partially ordered sets, shortly {\em posets},
which we represent as reflexive, antisymmetric and transitive digraphs.
Let the {\em width of a poset $P$} be the maximum size of an antichain in
$P$, that is the maximum independent set size in the digraph~$P$.

We formally define twin-width using the ``matrix-partitioning'' view
of~\cite[Section~5]{DBLP:conf/focs/Bonnet0TW20}, and we restrict ourselves only to
the symmetric twin-width which is relevant to graphs and posets.
Let $\mx A$ be a square matrix with entries from a finite set $L$
(e.g., $L=\{0,1\}$ for undirected graphs and $L=\{0,1,-1,2\}$ for digraphs),
and assume that both the rows and the columns of $\mx A$ are indexed by the
same ground set~$X$.
Let $\ca R$ denote any partition of $X$ into nonempty sets.
For two parts $R,Q\in\ca R$, the submatrix of $\mx A$ formed by the rows
indexed by $R$ and the columns indexed by $Q$ is called the
{\em($R\times Q$) zone} of~$\mx A$.
Naturally, a zone of $\mx A$ is {\em constant} if all entries in the zone are equal.
For $P\in\ca R$, the {\em error value} (the ``red degree'') of a row $P$
(column~$P$) in $\mx A$ is the number of non-constant zones $P\times Q$ 
(zones $Q\times P$, respectively) in $\mx A$ over all $Q\in\ca R$
(including~$Q=P$).

\begin{definition}\label{df:symtww}
Let $\mx A$ be a square matrix with the rows and columns indexed by a ground
set~$X$, $|X|=n$.
We say that the {\em symmetric twin-width of~$\mx A$ is at most~$d$}
if there exists a sequence of partitions $\ca R^1,\ldots,\ca R^n$ of $X$
(a {\em contraction sequence}) such that;
\begin{itemize}
\item $\ca R^1$ is the finest partition of $X$ ($|\ca R^1|=n$)
and $\ca R^n$ is the coarsest partition of $X$ ($|\ca R^n|=1$),
\item for each $i=1,\ldots,n-1$, the partition $\ca R^{i+1}$ results by
merging (``contraction'' of) some two parts of $\ca R^i$, and
\item for each $i=1,\ldots,n$ and every $P\in\ca R^i$, 
the error value of the row $P$ and the column $P$ in $\mx A$ is at most~$d$.
\end{itemize}
\end{definition}

For a quick illustration, consider Definition~\ref{df:symtww} applied to the
adjacency matrix $\mx A_G$ of a graph~$G$.
Then the definition of symmetric twin-width of $\mx A_G$ coincides with the twin-width of
$G$ as stated in Section~\ref{introduction}, except that the diagonal zones
($Q\times Q$ for $Q\in\ca R^i$ at step $i$) of $\mx A_G$ are often
non-constant, and hence the symmetric twin-width of $\mx A_G$ may be equal
or by one higher than the twin-width of $G$ (this difference is neglected in~\cite{DBLP:conf/focs/Bonnet0TW20}).

For a poset $P=(X,\leq_P)$, viewed as a digraph on the ground set~$X$,
we consider the matrix $\mx A_P$ defined as follows (according to \cite{DBLP:conf/focs/Bonnet0TW20});
$a_{u,v}=1$ iff $u\leq_P v$, $a_{u,v}=-1$ iff $v\lneq_P u$, and $a_{u,v}=0$ otherwise%
\footnote{If we considered general digraphs (which are not always antisymmetric), 
we would also consider value $a_{u,v}=2$ iff both $(u,v),(v,u)$ were edges of the digraph.}.
The {\em symmetric twin-width of $P$} is the symmetric twin-width~of~$\mx A_P$.

For our purpose of giving a fine relation between the width and the
twin-width of posets it is, though, much more convenient to use a
specialized definition which we call a {\em natural twin-width} of posets
(for a distinction).
We shortly write $a\sim_Pb$ iff $a\leq_Pb$ or $b\leq_Pa$ (i.e., the vertices
$a,b$ are comparable in~$P$).
Our definition reads:

\begin{definition}[Natural twin-width of a poset]\label{df:nattww}
A triple $P=(X,\leq_P,R)$ is a {\em red poset} if $P_0=(X,\leq_P)$ is a poset
and $R$ is a set of unordered pairs of incomparable elements of $P_0$.
The {\em red degree} of $P$ is the maximum degree of the ``red'' graph~$(X,R)$.

A {\em contraction} of two vertices $x_1,x_2\in X$ of $P$ (into a new
vertex~$x$) creates the red poset $P'=(X',\leq_P',R')$ 
where $X'=(X\setminus\{x_1,x_2\})\cup\{x\}$ and
\begin{itemize}
\item $a\leq_P'b$ iff $a\leq_Pb$ for all $a,b\in X\setminus\{x_1,x_2\}$, and $x\leq_P'x$,
\item $a\leq_P'x$ (resp.~$x\leq_P'a$) iff $a\leq_Px_1$ and $a\leq_Px_2$
(resp.~$x_1\leq_Pa$ and $x_2\leq_Pa$),
\item $R'=(R\parc X')\cup R_1\cup R_2$ where
$R_1=\{\{a,x\}:\{a,x_1\}\in R\vee \{a,x_2\}\in R\}$ and
$R_2=\{\{a,x\}: a\not\sim_P'x \>\wedge (a\sim_Px_1\vee a\sim_Px_2)\}$.
\end{itemize}
In other words, the red edges $R'$ of $P'$ are; (i) those inherited from $P$
which compose of the restriction of $R$ to $X'$ and the red edges $R_1$ formerly
incident to $x_1$ or $x_2$ in~$P$,
and (ii) the new ones in $R_2$ between $x$ and those vertices of $X\setminus\{x_1,x_2\}$
which compared in $P$ to $x_1$ and to $x_2$ in different ways.

An (ordinary) poset $P_0=(X,\leq_{P_0})$ has the {\em natural twin-width at
most~$d$} if the red poset $(X,\leq_{P_0},\emptyset)$ can be reduced down to a single
vertex by a sequence of contractions such that, at each step, the red degree
is at most~$d$.
So, the natural twin-width of $P_0$ equals the minimum integer $d$ such that
$P_0$ has the natural twin-width at most~$d$.
\end{definition}

With respect to the definition of a contraction (in red posets), the
following is a useful convention:
If a red poset $P=(X,\leq_P,R)$ resulted by a sequence of contractions from
a poset $P_1=(X_1,\leq_{P_1})$, then a vertex $y\in X$ uniquely corresponds
to a set $Y\subseteq X_1$ of those vertices of $P_1$ which were contracted
down to~$y$, and hence we will chiefly refer to $Y$ with the name 
$y\subseteq X_1$.
Consequently, the vertices of such $X$ at the same time form a partition of~$X_1$~in~$P_1$
(and, with negligible abuse of notation, $X_1$ itself is viewed as the
partition of $X_1$ into singletons)
which brings us very close to Definition~\ref{df:symtww} of symmetric twin-width.

In accordance with this convention, we will sometimes denote the vertex resulting by a
contraction of the vertices $x_1$ and $x_2$ shortly by $x_1x_2$.
See an example in Figure~\ref{fig:exampleseq}.

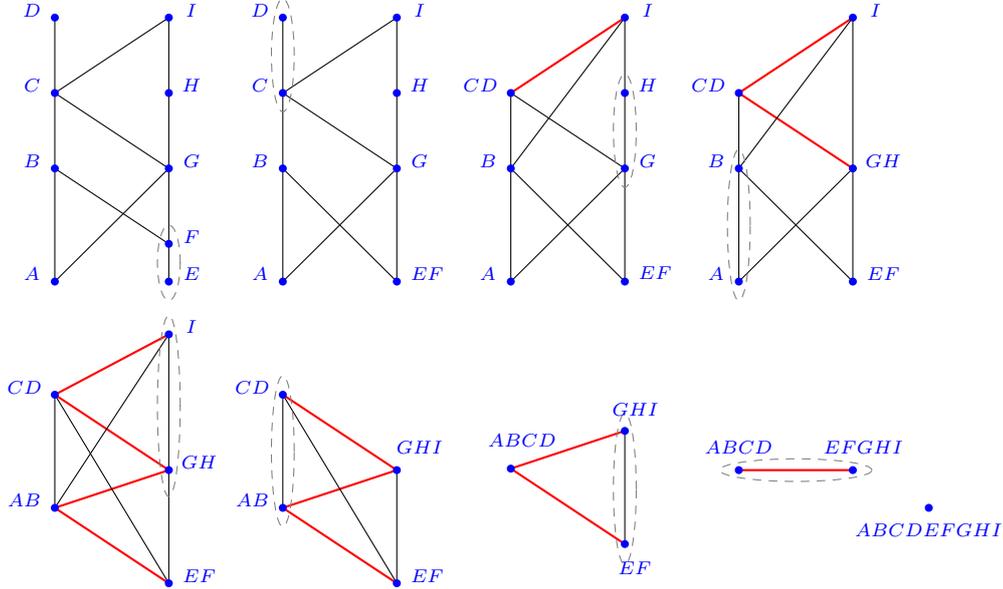
\begin{figure}[th]
$$
\begin{tikzpicture}[line cap=round,line join=round, x=1.0cm,y=1.0cm]
\draw (1,1)-- (1,2.5);
\draw (1,2.5)-- (1,3.5);
\draw (1,3.5)-- (1,4.5);
\draw (2.5,1)-- (2.5,1.5);
\draw (2.5,1.5)-- (2.5,2.5);
\draw (2.5,2.5)-- (2.5,3.5);
\draw (2.5,3.5)-- (2.5,4.5);
\draw (2.5,1.5)-- (1,2.5);
\draw (1,1)-- (2.5,2.5);
\draw (2.5,2.5)-- (1,3.5);
\draw (1,3.5)-- (2.5,4.5);
\draw (4,1)-- (4,2.5);
\draw (4,2.5)-- (4,3.5);
\draw (4,3.5)-- (4,4.5);
\draw (5.5,1)-- (5.5,2.5);
\draw (5.5,2.5)-- (5.5,3.5);
\draw (5.5,3.5)-- (5.5,4.5);
\draw (5.5,1)-- (4,2.5);
\draw (4,1)-- (5.5,2.5);
\draw (5.5,2.5)-- (4,3.5);
\draw (4,3.5)-- (5.5,4.5);
\draw (7,1)-- (7,2.5);
\draw (7,2.5)-- (7,3.5);
\draw (8.5,1)-- (8.5,2.5);
\draw (8.5,2.5)-- (8.5,3.5);
\draw (8.5,3.5)-- (8.5,4.5);
\draw (8.5,1)-- (7,2.5);
\draw (7,1)-- (8.5,2.5);
\draw (8.5,2.5)-- (7,3.5);
\draw [color=red, thick] (7,3.5)-- (8.5,4.5);
\draw (8.5,4.5)-- (7,2.5);
\draw (10,1)-- (10,2.5);
\draw (10,2.5)-- (10,3.5);
\draw (11.5,1)-- (11.5,2.5);
\draw (11.5,1)-- (10,2.5);
\draw (10,1)-- (11.5,2.5);
\draw [color=red, thick] (11.5,2.5)-- (10,3.5);
\draw [color=red, thick] (10,3.5)-- (11.5,4.5);
\draw (11.5,4.5)-- (10,2.5);
\draw (11.5,4.5)-- (11.5,2.5);
\draw (1,-2)-- (1,-0.5);
\draw (2.5,-3)-- (2.5,-1.5);
\draw [color=red, thick] (2.5,-3)-- (1,-2);
\draw [color=red, thick] (2.5,-1.5)-- (1,-0.5);
\draw [color=red, thick] (1,-0.5)-- (2.5,0.3);
\draw (2.5,0.3)-- (1,-2);
\draw (2.5,0.3)-- (2.5,-1.5);
\draw [color=red, thick] (1,-2)-- (2.5,-1.5);
\draw (4,-2)-- (4,-0.5);
\draw (5.5,-3)-- (5.5,-1.5);
\draw [color=red, thick] (5.5,-3)-- (4,-2);
\draw [color=red, thick] (5.5,-1.5)-- (4,-0.5);
\draw [color=red, thick] (4,-2)-- (5.5,-1.5);
\draw (2.5,-3)-- (1,-0.5);
\draw (5.5,-3)-- (4,-0.5);
\draw (8.5,-2.48)-- (8.5,-0.98);
\draw [color=red, thick] (8.5,-2.48)-- (7,-1.48);
\draw [color=red, thick] (7,-1.48)-- (8.5,-0.98);
\draw [color=red, thick] (11.5,-1.5)-- (10,-1.5);
\draw[color=gray,dashed] (2.5,1.25) ellipse (0.15 and 0.5);
\draw[color=gray,dashed] (4,4) ellipse (0.15 and 0.75);
\draw[color=gray,dashed] (8.5,3) ellipse (0.15 and 0.75);
\draw[color=gray,dashed] (10,1.75) ellipse (0.15 and 1);
\draw[color=gray,dashed] (2.5,-0.65) ellipse (0.15 and 1.2);
\draw[color=gray,dashed] (4,-1.25) ellipse (0.15 and 1);
\draw[color=gray,dashed] (8.5,-1.75) ellipse (0.15 and 1);
\draw[color=gray,dashed] (10.75,-1.5) ellipse (1 and 0.15);
\begin{scriptsize}
\fill [color=blue] (1,1) circle (1.5pt);
\draw[color=blue] (0.7,1.1) node {$A$};
\fill [color=blue] (1,2.5) circle (1.5pt);
\draw[color=blue] (0.7,2.6) node {$B$};
\fill [color=blue] (1,3.5) circle (1.5pt);
\draw[color=blue] (0.7,3.6) node {$C$};
\fill [color=blue] (1,4.5) circle (1.5pt);
\draw[color=blue] (0.7,4.6) node {$D$};
\fill [color=blue] (2.5,1) circle (1.5pt);
\draw[color=blue] (2.8,1.1) node {$E$};
\fill [color=blue] (2.5,1.5) circle (1.5pt);
\draw[color=blue] (2.8,1.6) node {$F$};
\fill [color=blue] (2.5,2.5) circle (1.5pt);
\draw[color=blue] (2.8,2.6) node {$G$};
\fill [color=blue] (2.5,3.5) circle (1.5pt);
\draw[color=blue] (2.8,3.6) node {$H$};
\fill [color=blue] (2.5,4.5) circle (1.5pt);
\draw[color=blue] (2.8,4.6) node {$I$};
\fill [color=blue] (4,1) circle (1.5pt);
\draw[color=blue] (3.7,1.1) node {$A$};
\fill [color=blue] (4,2.5) circle (1.5pt);
\draw[color=blue] (3.7,2.6) node {$B$};
\fill [color=blue] (4,3.5) circle (1.5pt);
\draw[color=blue] (3.7,3.6) node {$C$};
\fill [color=blue] (4,4.5) circle (1.5pt);
\draw[color=blue] (3.7,4.6) node {$D$};
\fill [color=blue] (5.5,1) circle (1.5pt);
\draw[color=blue] (5.9,1.1) node {$EF$};
\fill [color=blue] (5.5,2.5) circle (1.5pt);
\draw[color=blue] (5.8,2.6) node {$G$};
\fill [color=blue] (5.5,3.5) circle (1.5pt);
\draw[color=blue] (5.8,3.6) node {$H$};
\fill [color=blue] (5.5,4.5) circle (1.5pt);
\draw[color=blue] (5.8,4.6) node {$I$};
\fill [color=blue] (7,1) circle (1.5pt);
\draw[color=blue] (6.7,1.1) node {$A$};
\fill [color=blue] (7,2.5) circle (1.5pt);
\draw[color=blue] (6.7,2.6) node {$B$};
\fill [color=blue] (7,3.5) circle (1.5pt);
\draw[color=blue] (6.6,3.6) node {$CD$};
\fill [color=blue] (8.5,1) circle (1.5pt);
\draw[color=blue] (8.9,1.12) node {$EF$};
\fill [color=blue] (8.5,2.5) circle (1.5pt);
\draw[color=blue] (8.8,2.6) node {$G$};
\fill [color=blue] (8.5,3.5) circle (1.5pt);
\draw[color=blue] (8.8,3.6) node {$H$};
\fill [color=blue] (8.5,4.5) circle (1.5pt);
\draw[color=blue] (8.8,4.6) node {$I$};
\fill [color=blue] (10,1) circle (1.5pt);
\draw[color=blue] (9.7,1.1) node {$A$};
\fill [color=blue] (10,2.5) circle (1.5pt);
\draw[color=blue] (9.7,2.6) node {$B$};
\fill [color=blue] (10,3.5) circle (1.5pt);
\draw[color=blue] (9.6,3.6) node {$CD$};
\fill [color=blue] (11.5,1) circle (1.5pt);
\draw[color=blue] (11.9,1.1) node {$EF$};
\fill [color=blue] (11.5,2.5) circle (1.5pt);
\draw[color=blue] (11.9,2.6) node {$GH$};
\fill [color=blue] (11.5,4.5) circle (1.5pt);
\draw[color=blue] (11.8,4.6) node {$I$};
\fill [color=blue] (1,-2) circle (1.5pt);
\draw[color=blue] (0.6,-1.9) node {$AB$};
\fill [color=blue] (1,-0.5) circle (1.5pt);
\draw[color=blue] (0.6,-0.4) node {$CD$};
\fill [color=blue] (2.5,-3) circle (1.5pt);
\draw[color=blue] (2.9,-2.9) node {$EF$};
\fill [color=blue] (2.5,-1.5) circle (1.5pt);
\draw[color=blue] (2.9,-1.4) node {$GH$};
\fill [color=blue] (2.5,0.3) circle (1.5pt);
\draw[color=blue] (2.8,0.4) node {$I$};
\fill [color=blue] (4,-2) circle (1.5pt);
\draw[color=blue] (3.6,-1.9) node {$AB$};
\fill [color=blue] (4,-0.5) circle (1.5pt);
\draw[color=blue] (3.6,-0.4) node {$CD$};
\fill [color=blue] (5.5,-3) circle (1.5pt);
\draw[color=blue] (5.9,-2.9) node {$EF$};
\fill [color=blue] (5.5,-1.5) circle (1.5pt);
\draw[color=blue] (5.8,-1.2) node {$GHI$};
\fill [color=blue] (7,-1.48) circle (1.5pt);
\draw[color=blue] (7.15,-1.1) node {$ABCD$};
\fill [color=blue] (8.5,-2.48) circle (1.5pt);
\draw[color=blue] (8.63,-2.8) node {$EF$};
\fill [color=blue] (8.5,-0.98) circle (1.5pt);
\draw[color=blue] (8.64,-0.7) node {$GHI$};
\fill [color=blue] (10,-1.5) circle (1.5pt);
\draw[color=blue] (10,-1.2) node {$ABCD$};
\fill [color=blue] (11.5,-1.5) circle (1.5pt);
\draw[color=blue] (11.64,-1.2) node {$EFGHI$};
\fill [color=blue] (12.5,-2) circle (1.5pt);
\draw[color=blue] (12.5,-2.3) node {$ABCDEFGHI$};
\end{scriptsize}
\end{tikzpicture}
$$
\caption{An example of a contraction sequence for the top-left poset (each
step contracts the encircled pair), having the red degree at most~$2$.
As in Hasse diagrams, the edges (black) of the posets are oriented up, and we skip drawing edges
which are implied by reflexivity and transitivity.}
\label{fig:exampleseq}
\end{figure}

\begin{proposition}
If the symmetric twin-width of a poset $P$ is $d_s$ and the natural
twin-width of $P$ is $d$, then $d\leq d_s\leq d+1$.
\end{proposition}
\begin{proof}[Proof~(sketch)]
We compare Definitions \ref{df:symtww} and~\ref{df:nattww} for the
same contraction sequence; a non-constant zone corresponds to a created red
edge, and vice versa.
The only difference is at the diagonal zones which may be non-constant while
natural twin-width does not consider red loops, and so the symmetric
twin-width may be by one higher than the natural one.
\end{proof}

\subsection{Simple lower bound}

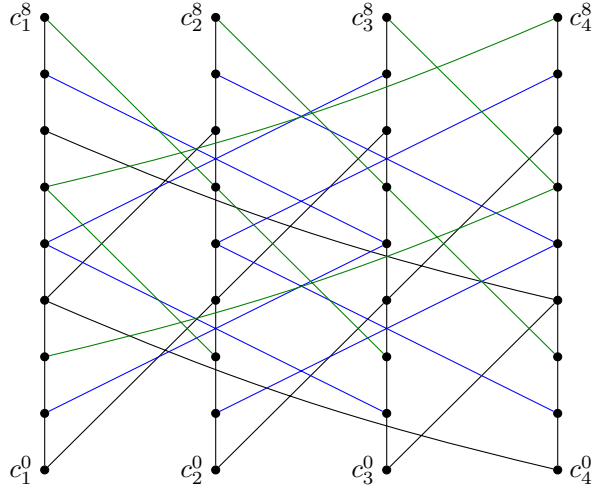
\begin{figure}[th]
$$
\begin{tikzpicture}[scale=0.75]
\tikzstyle{every node}=[draw, shape=circle, minimum size=3pt,inner sep=0pt, fill=black]
\node[label=left:$c_1^0$] at (0,0) (a0) {}; \node at (0,1) (a1) {}; \node at (0,2) (a2) {};
\node at (0,3) (a3) {}; \node at (0,4) (a4) {}; \node at (0,5) (a5) {}; 
\node at (0,6) (a6) {}; \node at (0,7) (a7) {}; \node[label=left:$c_1^8$] at (0,8) (a8) {}; 
\node[label=left:$c_2^0$] at (3,0) (b0) {}; \node at (3,1) (b1) {}; \node at (3,2) (b2) {};
\node at (3,3) (b3) {}; \node at (3,4) (b4) {}; \node at (3,5) (b5) {}; 
\node at (3,6) (b6) {}; \node at (3,7) (b7) {}; \node[label=left:$c_2^8$] at (3,8) (b8) {}; 
\node[label=left:$c_3^0$] at (6,0) (c0) {}; \node at (6,1) (c1) {}; \node at (6,2) (c2) {};
\node at (6,3) (c3) {}; \node at (6,4) (c4) {}; \node at (6,5) (c5) {}; 
\node at (6,6) (c6) {}; \node at (6,7) (c7) {}; \node[label=left:$c_3^8$] at (6,8) (c8) {}; 
\node[label=right:$c_4^0$] at (9,0) (d0) {}; \node at (9,1) (d1) {}; \node at (9,2) (d2) {};
\node at (9,3) (d3) {}; \node at (9,4) (d4) {}; \node at (9,5) (d5) {}; 
\node at (9,6) (d6) {}; \node at (9,7) (d7) {}; \node[label=right:$c_4^8$] at (9,8) (d8) {};
\draw (a0) -- (a8) (b0) -- (b8) (c0) -- (c8) (d0) -- (d8) ; 
\draw (a0) -- (b3) (b0) -- (c3) (c0) -- (d3) (d0) to[bend left=5] (a3) ;
\draw (a3) -- (b6) (b3) -- (c6) (c3) -- (d6) (d3) to[bend left=5] (a6) ;
\color{blue}
\draw (a1) -- (c4) (b1) -- (d4) (c1) -- (a4) (d1) -- (b4) ;
\draw (a4) -- (c7) (b4) -- (d7) (c4) -- (a7) (d4) -- (b7) ;
\color{green!50!black}
\draw (a2) to[bend right=5] (d5) (b2) -- (a5) (c2) -- (b5) (d2) -- (c5) ;
\draw (a5) to[bend right=5] (d8) (b5) -- (a8) (c5) -- (b8) (d5) -- (c8) ;
\end{tikzpicture}
$$
\caption{The poset from the proof of Proposition~\ref{pro:lowerd} for $d=4$
and $k=8$, depicted by its Hasse diagram.}
\label{fig:examplelowd}
\end{figure}

\begin{proposition}\label{pro:lowerd}
There exists a poset of width $d$ and the natural twin-width at least $d-1$.
\end{proposition}
\begin{proof}
For~$d=2$, we simply take a poset which has no pair of twin vertices,
and so any contraction in it creates a red edge, that is, red degree $d-1=1$.
E.g., we take the poset formed by the divisibility relation on the set
$\{2,3,4,6\}$.

For $d\geq3$ and $k\geq4d-8$, we construct a poset $P$ on a ground set of $n=d(k+1)$ vertices 
$C_1\cup C_2\cup\ldots\cup C_d$ where each $C_i=\{c_i^0,c_i^1,\ldots,c_i^k\}$ is a
chain; $c_i^0\leq_P c_i^1\leq_P \ldots\leq_P c_i^k$.
In the description of~$P$, we consider indices $i$ ``modulo $d$'', formally,
we set $c_{d+1}^j=c_1^j$, \dots, $c_{d+d}^j=c_{d}^j$.
Furthermore, for $i=1,\ldots,d$ and $j=0,\ldots,k-d+1$,
with \mbox{$a=1+\,j\!\mod(d-1)$}, we declare $c_i^{j}\leq_P c_{i+a}^{j+d-1}$.
The rest of $\leq_P$ follows by the reflexive and transitive closure.
See an example of this construction in Figure~\ref{fig:examplelowd},
and note that the inverse of $P$ is isomorphic to $P$
(informally, turning $P$ ``upside down'' gives the same poset)
which will be used to reduce the number of cases in the coming arguments by symmetry.

Our aim is to prove that a contraction of any pair of vertices of $P$
already gives red degree~$\geq d-1$.
Suppose first that the contracted pair is from the same chain~$C_i$, e.g.,
$c_i^j$ and $c_i^h$ for $0\leq j<h\leq k$.
We may assume $j\leq k-2d+3$, since otherwise we would have $h\geq j+1\geq k-2d+5\geq2d-3$
and could apply the symmetric argument.
Then, for \mbox{$a=1+\,j\!\mod(d-1)$},\, $c_i^{j}\leq_P c_{i+a}^{j+d-1}$ but
$c_i^{h}\not\leq_P c_{i+a}^{j+2d-3}$.
Therefore, the contracted vertex $c_i^{j}c_i^{h}$ has at least
$(j+2d-3)+1-(j+d-1)=d-1$ incident red edges to the chain~$C_{i+a}$.

Suppose now that the contracted pair is $c_i^j$ and $c_y^h$ where $i\not=y$
and $0\leq j\leq h\leq k$.
Again, we may assume by symmetry that $j\leq k-d+1$.
If $h>j$, then $c_y^h\not\leq_P c_i^{j+d-1}$ and the contracted vertex
$c_y^hc_i^j$ has at least $d-1$ incident red edges to the chain~$C_{i}$.
If $h=j$, then $c_y^hc_i^j$ similarly has at least $2(d-2)\geq d-1$ 
incident red edges to the chains~$C_{i}$ and~$C_y$.
\end{proof}

\section{Upper bound for posets of width \boldmath$d$}
\label{sec:twwd}

Complementing Proposition~\ref{pro:lowerd}, we give the core upper
estimate followed by its proof:
\begin{theorem}\label{thm:maind}
A poset of width $d$ has the natural twin-width at most $9d-6$,
and hence the symmetric (matrix) twin-width at most~$9d-5$.
The corresponding contraction sequence can be found in time
$\ca O(dn^2)$ where $n$ is the number of vertices of the poset.
\end{theorem}

By Dilworth's theorem, a poset $P=(X,\leq_P)$ is of width $d$ if and only if
the ground set $X$ can be partitioned into at most $d$ {\em chains} (a chain
is linearly ordered by~$\leq_P$).
Hence, from now on, we will consider a poset of width $d$ with a fixed
partition $\pi$ of $X$ into $d$ (nonempty) chains, formally as a triple
$P=(X,\leq_P,\pi)$ where $\pi=\{U_1,\ldots,U_d\}$.

Our upper bound in Theorem~\ref{thm:maind} will use only a special type of
contractions -- of two consecutive vertices of the same chain of $P$.
Since contractions inside a chain essentially preserve the chain partition of $P$, 
we will for simplicity refer to the new chain partition as to $\pi$ again.
We shall thus work with the following special kind of red posets, which
result from a chain-partitioned poset by our special contractions:

\begin{definition}[red $d$-neighbourly poset]\label{df:redneighbposet}
Let $P_0=(X_0,\leq_{P_0},\pi)$ be a poset partitioned by $\pi$ into $d$ chains.
A {\em neighbourly contraction} is a contraction of a vertex pair $x_1,x_2$
such that $x_1\not=x_2$ belong to the same chain $U$ of $\pi$ and they are
consecutive in this chain (i.e., no element of $U$ is strictly between $x_1$ and $x_2$).
Note, however, that for such a pair $x_1,x_2$ there could exist a vertex $y$
in another chain of $\pi$ such that $x_1\lneq y\lneq x_2$.

A tuple $P=(X,\leq_P,\pi,R)$ is called a {\em red $d$-neighbourly
poset} (shortly a {\em neighbourly poset}) if the red poset $(X,\leq_P,R)$ is
obtained from $P_0$ by an arbitrary sequence of neighbourly contractions
(we shortly say that $P$ is a {\em contraction of~$P_0$}).
\end{definition}

Roughly speaking, our proof of Theorem~\ref{thm:maind} is going to argue
that, although some neighbourly contractions can create many new red edges,
overall the number of red edges that can potentially be created by every
neighbourly contraction is only proportional to the size of the poset.
Therefore, one can always find a ``good'' neighbourly contraction.
Later on in the contraction sequence, we also have to watch the number 
of previously created red edges which is straightforward.
Of course, to fulfill Definition~\ref{df:nattww}, we will have to contract
the remaining $d$ single-vertex chains at the end together, but that part will be a
trivial conclusion of our proof, and so we neglect it in the coming technical arguments.

When dealing with neighbourly posets such as $P=(X,\leq_P,\pi,R)$, 
we adopt some special notation.
Consider a chain $U=(u_1,u_2,\ldots,u_m)$ ordered as $u_1\leq_P\ldots,\leq_Pu_m$.
For a vertex $x\in X$ such that $x=u_a$ of $U$, we shall write 
$x^{+i}$ for the vertex $u_{a+i}$ and $x^{-i}$ for the vertex $u_{a-i}$
(of course, assuming $a+i\leq m$ or $a-i\geq 1$, respectively).
Let $x^+$ and $x^-$ be a shorthand for $x^{+1}$ and $x^{-1}$.

We give a unified way of picturing neighbourly posets
-- a {\em chain diagram} (already seen in Figure~\ref{fig:exampleseq}), 
which is close to the traditional Hasse diagram of a poset, but not exactly~the~same.
\begin{definition}[chain diagram]\label{df:chaindia}
Let $P$ be a red $d$-neighbourly poset.
Every chain of $P$ is drawn as a vertical line,
the red edges of $R$ are drawn as {\em red bars} between pairs of the chains, and
there is a {\em black bar} from a vertex $u$ of a chain $U$ to a vertex $v$
of a chain $V\not=U$, if and only if $v$ is the least vertex of $V$ greater than $u$
and $u$ is the greatest vertex of $U$ smaller than~$v$.%
\footnote{ Notice that since $R$ contains only incomparable pairs, there cannot be a red
and a black bar together between the same pair.}
A black bar is never drawn as horizontal and is implicitly directed up in the picture.
\end{definition}

\subsection{Structure of black and red bars}

While black bars of a neighbourly poset $P$ are directed by $\leq_P$,
that is we have a black bar $(u,v)$ if $u\leq_Pv$ as in Definition~\ref{df:chaindia},
red bars are by Definition~\ref{df:nattww} undirected.
Nevertheless, we can assign a direction to a red bar $\{u,v\}\in R$ as follows.

\begin{definition}[orienting the red bars]\label{df:redorient}
Let the neighbourly poset $P=(X,\leq_P,\pi,R)$ be a contraction of an ordinary poset
$P_0=(X_0,\leq_{P_0},\pi)$, and recall that we view $u\in X$ as the subset
$u\subseteq X_0$ of the respective contracted vertices of~$P_0$.
Since $u\subseteq X_0$ belongs to one chain of $P_0$ by Definition~\ref{df:redneighbposet},
the minimum $min(u)$ of $u$ is well-defined.

We orient the red bar $\{u,v\}\in R$ as $(u,v)$, from $u$ to~$v$, 
if $min(u)\leq_{P_0}x$ for some $x\in v$, but
$min(v)\not\leq_{P_0}min(u)$.
Though, if both directions $(u,v)$ and $(v,u)$ are assigned by this
criterion (which is possible, e.g., when the minima of $u$ and $v$ are
incomparable), then we choose $(u,v)$ if the last contraction into $u$
happened later than that into~$v$.%
\footnote{The latter criterion of choosing between $(u,v)$ and $(v,u)$ is 
not really important; we introduce it only to ``break the tie'' in a deterministic way.}
\end{definition}

Observe that at least one of the options for $\{u,v\}\in R$ in
Definition~\ref{df:redorient} must happen.
As an informal explanation, a red bar $\{u,v\}\in R$ in $P$ means that between the
sets $u,v\subseteq X_0$ in $P_0$, the edges (and non-edges) are not uniform
(not all in one direction),
and then we choose the ``prevailing direction'' for the orientation of $\{u,v\}$.
We shall write a red bar as $\{u,v\}\in R$ if we do not care about the
orientation of it, and as $(u,v)\in R$ if we~do~care.

Now we summarize basic technical properties used in further proofs.

\begin{lemma}\label{lem:bars}
Let $P=(X,\leq_P,\pi,R)$ be a $d$-neighbourly poset
and $U$ and $V$ be two distinct chains of $P$ determined by~$\pi$.
\begin{enumerate}[a)]
\item \label{it:blackup}
If $u\leq_Pv$ where $u\in U$ and $v\in V$, then there is no $i>0$ such that $(u,v^{+i})\in R$.
Analogously, if $u\geq_Pv$, then there is no $i>0$ such that $(v^{-i},u)\in R$.
\item \label{it:blackred}
If $(u,v)\in R$ where $u\in U$ and $v\in V$ such that $(u,v^+)\not\in R$, then $u\leq_{P}v^+$
(informally, red bars oriented from $u$ to vertices of $V\!$
come in a consecutive strip ``capped'' by a black~bar).
\\An analogous claim symmetrically holds for red bars oriented towards $u$ from $V$.
\item \label{it:rednocross}
If $(u,v)\in R$ where $u\in U$ and $v\in V$, then there are no $i,j>0$ such that 
$(u^{+i},v^{-j})\in R$ or $u^{+i}\leq_Pv^{-j}$
(informally, no two red bars of the same orientation from $U$ to $V$ may~``cross'', 
and no black bar from $U$ to $V$ may be ``crossed'' by a red bar starting below it~in~$U$).
\item \label{it:numredbars}
There are together at most $|U|+|V|$ red bars from a vertex of $U$ to
a vertex of~$V$.
\item \label{it:oneredto}
Assume $u\in U$, $v\in V$ such that $(u,v)\in R$ is a red bar that has been
newly created by a neighbourly contraction into $u$, 
and that no contraction into $v$ has happened after the creation of red $(u,v)$.
Then no neighbourly contraction in $P$ in the chain $U$ 
can create another red bar oriented towards~$v$.
\\An analogous claim holds for $(v,u)\in R$ and creation of red bars oriented from~$v$.
\end{enumerate}
\end{lemma}

\begin{proof}
Let $P$ be a contraction of the ordinary poset $P_0=(X_0,\leq_{P_0},\pi)$.

\ref{it:blackup})
Trivially by transitivity, $u\leq_Pv\leq_Pv^{+i}$ contradicts that pairs in
$R$ are incomparable.

\smallskip\ref{it:blackred})
By $(u,v)\in R$ and Definition~\ref{df:redorient}, for some $x\in v$ of $P_0$ 
we have $min(u)\leq_{P_0} x\leq_{P_0} min(v^+)$.
Then, if $max(u)\not\leq_{P_0} min(v^+)$, we would have forbidden $(u,v^+)\in R$.
Therefore, by homogeneity of the edges from $u$ to $v^+$ in $P_0$,
we get desired $u\leq_{P}v^+$.

\smallskip\ref{it:rednocross})
Assume the contrary, that $(u^{+i},v^{-j})\in R$.
Then, by Definition~\ref{df:redorient} and transitivity in $P_0$,
$max(u)\leq_{P_0} min(u^{+i})\leq_{P_0} x\leq_{P_0} min(v)$ where
$x\in v^{-j}$, which contradicts the assumption~$\{u,v\}\in R$.
The same argument goes through if $u^{+i}\leq_Pv^{-j}$ with $x=min(v^{-j})$.

\smallskip\ref{it:numredbars})
We ignore the chains other than $U$ or~$V$.
Let $u_1$ be the minimum of the chain $U$, and assume that there are $q\geq0$
red bars oriented from $u_1$ to $V$.
Let $V_1\subseteq V$ be the lowest $q-1$ vertices of the chain~$V$.
By \eqref{it:rednocross}, only red bars from $u_1$ may end in $V_1$.
So, we remove the vertices $\{u_1\}\cup V_1$ and, by induction,
there are at most $|U|-1+|V|-(q-1)=|U|+|V|-q$ red bars from
$U\setminus\{u_1\}$ to $V\setminus V_1$.
With the $q$ red bars starting in $u_1$ we get the desired~bound.

\smallskip\ref{it:oneredto})
If $(u,v)\in R$ has been created by a contraction into $u$,
and not by a prior contraction into~$v$, then
$min(u)\leq_{P_0}min(v)$ using Definition~\ref{df:nattww}.
Hence a further neighbourly contraction in the chain $U$ below $u$
cannot at all create a red bar incident to~$v$,
and a neighbourly contraction in the chain $U$ above $u$ can only create a
new red bar oriented from $v$,
according to Definition~\ref{df:redorient}.
\end{proof}

\subsection{Minimizing the red potential}

Now comes the core of the proof of Theorem~\ref{thm:maind}, estimating how many red edges can
potentially result from all possible neighbourly contractions in~$P$.
Let $u\in U$ be a vertex of a chain $U$ which is not maximal,
$u^0$ be the vertex created by the contraction of $u$ and $u^+$,
and define the {\em red potential of $u$} as the number of red edges
incident to $u^0$ after the contraction (so, previous red edges incident to
$u$ or $u^+$ are also counted here).
The {\em red potential of the chain $U$} is simply the sum of red potentials
over the non-max vertices of $U$, and the {\em red potential of $P$} is the
sum over all chains of $P$.

\begin{lemma}\label{lem:redpot}
\begin{enumerate}[a)]
\item If $P=(X,\leq_P,\pi,R)$ is a red $d$-neighbourly poset with $m=|X|$ elements,
then the red potential of $P$ is at most $2(d-1)m+2|R|$.
\item There are at most $|R|\leq2(d-1)m$ red bars in~$P$.
\end{enumerate}
\end{lemma}

\begin{proof}
a)
Consider one chain $U$ of~$P$ and the set $R_U\subseteq R$ of red
bars incident with a vertex of~$U$.
We show that the red potential of $U$ is at most $2(m-|U|)+|R_U|$.

Regarding the part of red potential contributed by new red bars
(i.e., not those inherited from $R_U$), this follows easily.
By Definition~\ref{df:redorient}, for every  $x\in X\setminus U$, 
at most one neighbourly contraction in $U$ creates a
new red bar oriented from $U$ to $x$ and at most one such oriented from~$x$;
hence the term $2(m-|U|)$ in the estimate.
In addition to the previous,
every red bar $\{u,v\}\in R_U$ where $u\in U$ is inherited by (i.e., contributes $+1$
to) at most two neighbourly contractions in $U$; namely to those of the pairs
$u^-,u$ and $u,u^+$.
However, we are double-counting this way, and we now show that it is enough 
to count the ``$+1$'' contribution towards one of the two contractions.
If, up to symmetry, $(u,v)\in R_U$, we contribute it to the contraction
of the pair $u,u^+$ since the following holds:
If $(u^-,v)\in R_U$, then the contraction of $u^-,u$ into $u^0$ anyway makes only one
inherited red bar $(u^0,v)$ and this has been contributed to $u^-$ by our rule.
Otherwise, by Lemma~\ref{lem:bars}\eqref{it:blackred}, we have
$u^-\leq_Pv$ and the inherited red bar $(u^0,v)$ is the same as the
potential new red bar counted in the first~part~for~$x=v$.

Summing previous $2(m-|U|)+|R_U|$ over all $d$ chains $U$ of cardinalities
$m_1,m_2,\ldots,m_d$ in $P$, we count every red bar of $R$ exactly twice,
and this leads to desired
$$
\sum\nolimits_{1\leq i\leq d}2(m-m_i) +2|R| = 2dm-2\sum\nolimits_{1\leq i\leq d}m_i
 +2|R| = 2(d-1)m +2|R|
\,.$$

\smallskip b)
Again, let $m_1,m_2,\ldots,m_d$ be the cardinalities of the chains of $P$.
Counting the red bars over all ordered pairs of the chains,
using Lemma~\ref{lem:bars}\eqref{it:numredbars}, we get an upper bound of
$$
\sum\nolimits_{1\leq i\not=j\leq d}(m_i+m_j)
 =\sum\nolimits_{1\leq i\leq d}(d-1)m_i + \sum\nolimits_{1\leq j\leq d}(d-1)m_j
  = 2(d-1)m
$$
red bars in~$P$.
\end{proof}

\begin{proof}[Proof of Theorem~\ref{thm:maind}]
Let $P_0=(X_0,\leq_{P_0},\pi)$ be an ordinary poset partitioned into $d$ chains.
We are now ready to finish the main proof; to find a desired contraction
sequence of $P_0$ of bounded red degree.
The natural idea at each step (with a neighbourly poset~$P$ obtained from
$P_0$ so far) is to exhaustively find a neighbourly contraction in $P$ 
of the smallest red potential, which is
upper-bounded independently of the size of $P$ based on Lemma~\ref{lem:redpot}.

Let $P=(X,\leq_P,\pi,R)$ be a contraction of $P_0$, as above, and $m=|X|\leq|X_0|=n$.
The red potential of whole $P$ is at most
$2(d-1)m+4(d-1)m=6(d-1)m$ by Lemma~\ref{lem:redpot}.
Since there are $m-d$ possible neighbourly contractions in $P$, one of them
has the red potential at most $\frac{6(d-1)m}{m-d}$.
If $m\geq7d$, we have $m-d\geq\frac67m$ and then the red potential
$\frac{6(d-1)m}{m-d}\leq\frac{7(d-1)m}{m}=7(d-1)$.
Otherwise, $m<7d$ and then the red degree in $P$ is anyway at most
$m-1\leq7d-2\leq9d-6$, and so we finish with any sequence of contractions.

We are nearly done, but there is a small catch.
For a vertex $x\in X$ of $P$, call a red edge $\{x,y\}$ incident to $x$ {\em domestic}
(to~$x$) if it has been there already the last time we have contracted into $x$ along
our contraction sequence; otherwise, call red $\{x,y\}$ {\em foreign}.
While the argument in the previous paragraph bounded the number of domestic red edges
incident to any vertex along the whole sequence, we have not yet bounded the
number of potential foreign red edges (that is those which have been created
by contraction to other vertices later~on).
Using Lemma~\ref{lem:bars}\eqref{it:oneredto}, we argue that there can be at
most one foreign edge incident to $x$ oriented towards $x$, and one oriented
from $x$, per each other chain of $P$.
Hence the number of foreign red edges incident to any $x$ is at most
$2(d-1)$, and the maximum red degree along our contraction sequence thus is
at most $7(d-1)+2(d-1)<9d-6$.

The above proof straightforwardly translates into a simple and efficient algorithm.
The red potential of one vertex of $P$ can be found in time proportional to
$d$ and the value of this red potential, and hence the minimum red potential 
of the poset $P$ in the current step of a contraction sequence is determined
in time $\ca O(dm)\leq\ca O(dn)$.
The same time is sufficient to update $P$ for the next step.
Since we need $n-1$ steps of the contraction sequence for $P_0$,
this computation is finished in time $\ca O(dn^2)$.
\end{proof}

\section{Tight estimate for posets of width 2}
\label{sec:tww2}

From Section~\ref{sec:twwd} we get that in the worst case a poset of width
$2$ has twin-width at least $1$ and at most $12$ (where the upper bound of
Theorem~\ref{thm:maind} can likely be improved a bit in this special case of $d=2$).
However, if one wants to get an exact worst-case value of twin-width,
namely value~$2$, a very different approach needs to be employed,
one which is special only for posets of width $2$ and does not generalize
even to width~$3$.

We start with the upper estimate:

\begin{theorem}\label{thm:main2}
A poset of width $2$ has the natural twin-width at most $2$,
and the corresponding contraction sequence can be computed in linear time.
\end{theorem}

We prove the statement by providing the claimed algorithm and proving its correctness.
On a high level, our algorithm performs a depth-first search for a ``safe''
possibility of a neighbourly contraction in one of the two chains of a poset~$P$,
starting from a minimal vertex $u_1$ of the poset.
By a safe neighbourly contraction we mean one in which we have or create at
most one incoming and at most one outgoing red bar in the contracted vertex.
We also stay in firm control of all red bars in intermediate contracted
red posets.

To control the search for neighbourly contraction pairs and the occurrence
of red bars, we introduce the notion of a {\em directed bar path}
(recall also Definition~\ref{df:chaindia} of a chain diagram with bars).
A bar path in a red $2$-neighbourly poset $P$ is a directed path
represented as a vertex sequence $(x_1,x_2,\ldots,x_k)$ in $P$ such that
\begin{itemize}
\item $x_1$ is a minimal vertex of $P$ (either of the possible two),
\item for $i=1,\ldots,k-1$, $(x_i,x_{i+1})$ is a black or red bar in $P$
oriented this way, and
\item if both $x_i^+$,\,$x_{i+1}^+$ exist in $P$, then
	$x_i^+\not\leq_Px_{i+1}^+$.
\end{itemize}
Notice that a bar path is ``zig-zag'' switching between the two chains of~$P$.
Our bar-path controlled algorithm is then formalized in Algorithm~\ref{alg:width2}.

Commenting on this algorithm, we remark that the main part (the one
searching for a ``safe'' neighbourly contraction along bar path~$B$) 
is presented on lines \ref{it:startmain} to \ref{it:shortenB}.
The preceding supplementary part on lines \ref{it:ifBV} to \ref{it:contractV}
is there to prepare for a possible contraction at the root $u_1$ of bar path $B$;
if $u_1$ is not the global minimum of $P$, then the lower vertices
of the other chain $V$ of $P$ could create
many red bars oriented towards $u_1$ after neighbourly contraction to~$u_1$.
This is eliminated by safe contractions of the problematic vertices of $V$ to
$v_1$ on line \ref{it:contractV}, which make future neighbourly contractions to $u_1$
safe as well.

Actually, the course of Algorithm~\ref{alg:width2} is illustrated in previous
Figure~\ref{fig:exampleseq}.
It is $u_1=A$, $v_1=E$ (after the supplementary first step, $v_1=EF$),
and the bar path leading to the first contraction (of $C,D$) on line
\ref{it:contruu} is $B=(A,G,C)$.
Further on, for example, in the fourth picture (the top-right red poset)
we have $B=(A)$ and $B^+=(A,GH,CD,I)$.
In the fifth picture (bottom-left), we get $B=(AB,GH)$
and $B^+=(EF,AB,GH,CD,I)$, and so on.

\begin{algorithm}[ht]
\caption{Constructing a contraction sequence of red degree~$2$.}
\label{alg:width2}
\begin{algorithmic}[1]
\Require Given a poset $P_0=(X_0,\leq_{P_0},\pi)$ partitioned into $2$ chains.
\Ensure A record of the constructed contraction sequence of~$P_0$ of red degree~$\leq2$.
\medskip

\State  \textbf{declare} $P=(X,\leq_P,\pi,R)$~ a red $2$-neighbourly poset

\State  \textbf{declare} $B$~ a directed bar path (of black and red bars of~$P$)
\smallskip

\State $P$ $\gets$ red neighbourly poset $P=(X_0,\leq_{P_0},\pi,\emptyset)$;
\State $B$ $\gets$ bar path $(u_1)$,\, where $u_1$ is a minimal vertex of $P_0$ (any of possible two);
\State{\it// bar path $B$ directs the course of the algorithm (kind of
	similarly to classical DFS);
	note that $B$ stays rooted in $u_1=min(U)$ till the end of the main loop}
\smallskip
\smallskip
\While{$B\not=\emptyset$ }\label{it:whileB}
  \State $u$ $\gets$ the (upper) end vertex of bar path $B$, i.e., $B=(x_1,\ldots,x_k,u)$;
  \State $\{U,V\}$ $\gets$ the two chains of $P$ (by~$\pi$) such that $u\in U$;
\smallskip
  \If{$B=(u)$, which is equivalent to $u=min(U)$}
	\label{it:ifBV}
    \State{\it// here we contract the lower section of $V$ which is
	homogeneous towards $U$:}
    \State $v_1$ $\gets$ $min(V)$;
	~$u_2$ $\gets$ smallest $u_2\in U$ such that~$v_1\leq_Pu_2$, or
	$u_2$ nonexistent;
    \While{$|V|>1$ and $\big($$|U|=1$ or ($u_2$ exists and $v_1^+\leq_Pu_2$)$\big)$}
	\label{it:whileV}
      \State $P$ $\gets$ \textbf{contraction} of the pair $v_1$,\,$v_1^+$ in $P$;
      ~$v_1$ $\gets$ $v_1v_1^+$;
	\label{it:contractV}
    \EndWhile{}
  \EndIf{}
\smallskip
\State{\it// the main part; prolonging~$B$, or possibly contracting at the end 
	and shortening~$B$:}
  \State $v$ $\gets$ the smallest vertex $v\in V$ such that 
	$u\leq_Pv$ or $(u,v)\in R$, or $v$ nonexistent;
	\label{it:startmain}
\smallskip
  \If{$v$,\,$v^+$ and $u^+$ exist in $P$, and $u^+\not\leq_Pv^+$}
	\label{it:ifmain}
      \State{\it\quad~~// note that since $u^+\not\leq_Pv^+$,
	we get that $(u,v)\in R$ or $(u,v)$ is a black~bar}
      \State $B$ $\gets$ $(B,\,v)$, i.e., prolongation of bar path $B$ by the bar~$(u,v)$;
	\label{it:prolong}
  \Else
    \If{$u^+$ exists in $P$}
      \State{\it// the contraction here is safe; even if $v$,\,$v^+$ exist,
	we have $u^+\leq_Pv^+$ and 
	\hspace*{\algorithmicindent}\hspace*{\algorithmicindent}\hspace*{\algorithmicindent}
	no red bar towards~$v^+$ (or further up) is created}
      \State $P$ $\gets$ \textbf{contraction} of the pair $u$,\,$u^+$ in $P$;
	~$u$ $\gets$ $uu^+$;
	\label{it:contruu}
    \EndIf{}
    \If{$B=(u)$}
      \If{$u^+$ nonexistent in $P$}
	~$B$ $\gets$ $\emptyset$;
      \EndIf{}
    \Else
      \State $B$ $\gets$ $B\setminus(u)$, i.e., shortening of bar path $B$ by the
	last bar towards~$u$;
	\label{it:shortenB}
    \EndIf{}
  \EndIf{}

\EndWhile{}
\smallskip

\State $P$ $\gets$ \textbf{contraction} of the remaining pair $min(U)$,\,$min(V)$ in~$P$;
\end{algorithmic}
\end{algorithm}

\begin{proof}[Proof of Theorem~\ref{thm:main2}]
We refer to Algorithm~\ref{alg:width2} and the definition of bar path~$B$.
Let $P_0$ be the input ordinary poset and $P$ the current $2$-neighbourly
poset, as in the algorithm.
Let $U_0,V_0$ denote the two chains of $P_0$ and $u_1=min(U_0)$ be the start 
(root) of bar path~$B$ which stays fixed during the course of computation.
Let $v_1=min(V_0)$.
Note that we slightly abuse notation by referring to these vertices as to
$u_1$ and $v_1$ also in the poset $P$, after possible contractions into
$u_1$ or $v_1$.

For the purpose of analysis of Algorithm~\ref{alg:width2}, we define an
{\em extended bar path} $B^+\supseteq B$ of the current bar path $B$ in $P$ as follows:
$B^+$ starts with $(v_1,u_1)$ if $(v_1,u_1)\in R$, and $B^+$ starts in $u_1$ otherwise.
Then $B^+$ contains all bars of $B$ in order, then possibly one black bar
starting in the last vertex of $B$, and finally, $B^+$ ends
as a directed path using only red bars of~$P$.
We claim the following {\em invariant} at the beginning
and after every iteration of the loop from~line~\ref{it:whileB}:
\begin{enumerate}[(I)]
\item $B$ conforms to the conditions of a bar path.
\label{it:baruplus}
\item There exists an extended bar path $B^+$ of $B$ in $P$ containing all red bars of~$P$.
\label{it:ballred}
\end{enumerate}
To better understand the role of an extended bar path, observe that,
modulo renaming of contracted vertices, $B^+$ coincides with the former bar path
in the iteration in which the upper-most red bar of $B^+$ has been created
by a contraction.

\smallskip
Since an extended bar path is acyclic and does not repeat vertices
-- this is not trivial but follows from Definition~\ref{df:redorient} --
condition \eqref{it:ballred} then implies that the red degree of $P$ after
every iteration is at most~$2$, thus proving the conclusion of
Theorem~\ref{thm:main2}.
Our aim hence is to prove the invariant, by induction on the iterations of
the main cycle.

At the beginning, $B=(u_1)$ satisfies \eqref{it:baruplus} and
\eqref{it:ballred} is trivial since~$R=\emptyset$.
We now assume that these hold when an iteration of the main cycle (line \ref{it:whileB}) starts.
Possible contractions on line \ref{it:contractV} do not create red edges
or change $B$, except that when $|U|=1$ they could create the red bar
$(v_1,u_1)$ which will be included in our extended bar path there.

For the next arguments, note that line \ref{it:startmain}
always selects a red bar $(u,v)$ if there is one starting from $u$.
So, if there is no red bar from $u$ and $u^+\leq_Pv$, 
then $(u,v)$ is {\em not} a black bar either (Definition~\ref{df:chaindia})
and the extended bar path $B^+$   
from \eqref{it:ballred} before this iteration cannot reach $v$.
Consequently, the contraction of $u,\,u^+$ on line \ref{it:contruu} of this
iteration does not make $(u,v)$ red, and \eqref{it:baruplus} and
\eqref{it:ballred} are satisfied after the iteration with the same~$B^+$.
Hence we can assume that $(u,v)\in R$ or $u^+\not\leq_Pv$ hold in the
coming case analysis.

\smallskip
We now analyze the remaining cases according to the `if' statements 
from line \ref{it:ifmain} onward:
\begin{itemize}
\item 
If $v$,\,$v^+$ and $u^+$ exist and $u^+\not\leq_Pv^+$,
then $(u,v)$ is a red or black bar in $P$, and so we explicitly satisfy all
three conditions of a bar path for the prolongation $(B,v)$ on line \ref{it:prolong}.
No new red bars are created in this iteration, and we claim that an
extended bar path $B^+$ from the previous iteration contains or will contain the new bar
$(u,v)$; this is trivial if $(u,v)$ is red, and for a black bar $(u,v)$
it is the only bar of $P$ starting in $u$ anyway.
\item 
If $u^+$ and $v$ exist, and $v^+$ is nonexistent or $u^+\leq_Pv^+$,
then the contraction of $u,\,u^+$ on line \ref{it:contruu} makes $(u,v)$ red
if it was not such before.
No other red bar exists or is created from $u$ (although, the bar of $B$
towards $u$ may already be red).
Altogether, we inherit $B^+$ from the previous step,
and with (now) red $(u,v)$ this will be \eqref{it:ballred} a valid extended bar path 
of the shortened bar path $B\setminus(u)$ at the end of the iteration.
Not to forget \eqref{it:baruplus}, $B\setminus(u)$ will be a valid bar path as well.
\item 
If $B=(u)$ in the previous case, we do not shorten $B$ (since $u^+$ exists),
but the arguments stay the same.
\item 
if $u^+$ exists but $v$ does not, then an extended bar path cannot reach
beyond $u$. So, after the contraction and shortening $B$ by $u$ we will
trivially satisfy \eqref{it:baruplus} and \eqref{it:ballred} again.
\item 
Lastly, if $u^+$ is nonexistent,
then we only shorten $B$ by $u$, or we are at the end of the algorithm if
$u$ is the root of~$B$ (note that here the procedure on lines \ref{it:ifBV}
to \ref{it:contractV} also takes part and shortens $V$ to one vertex before
we stop the cycle from line \ref{it:whileB}).
\end{itemize}

It remains to argue why the algorithm stops, and what is the runtime.
The first part is clear since every iteration of the main loop either
prolongs the current bar path (which is bounded), or eventually finds a
next contraction pair.

As for the runtime, we use the following special representation of the
working poset $P$, which extends the chain diagram of Definition~\ref{df:chaindia}:
For every $u\in U$ we record, besides the red bars of $u$, the least 
$v\in V$ such that $u\leq_Pv$.
We input the poset $P_0$ as a traditional Hasse diagram, and we prepare our
representation of it in linear time with respect to~$|X_0|$.
The total number of iterations of the main cycle is linearly proportional to the
number of performed contractions in the main part.
Then, at every iteration of the main cycle, we can perform the computation,
and the update of the structure representing $P$, in constant time.
The only exception is a possible iteration of the inner cycle on line
\ref{it:whileV}, which is counted amortized towards the total number of
contractions.
\end{proof}

The lower estimate matching Theorem~\ref{thm:main2} is as follows.

\begin{figure}[tbh]
$$
\begin{tikzpicture}[scale=0.66]
\tikzstyle{every node}=[draw, shape=circle, minimum size=3pt,inner sep=0pt,
		color=blue, fill=blue]
\node[label=left:$a_1$] at (0,0) (a1) {};
\node[label=left:$a_2$] at (0,1) (a2) {};
\node[label=left:$a_3$] at (0,2) (a3) {};
\node[label=left:$a_4$] at (0,3) (a4) {};
\node[label=left:$a_5$] at (0,4) (a5) {};
\node[label=right:$b_1$] at (3,0) (b1) {};
\node[label=right:$b_2$] at (3,1) (b2) {};
\node[label=right:$b_3$] at (3,2) (b3) {};
\node[label=right:$b_4$] at (3,3) (b4) {};
\node[label=right:$b_5$] at (3,4) (b5) {};
\draw (a1) -- (a5) (b1) -- (b5);
\draw (a1) -- (b3) (b1) -- (a3);
\draw (a3) -- (b5) (b3) -- (a5);
\end{tikzpicture}
\medskip$$ $$
\begin{tikzpicture}[scale=0.56]
\tikzstyle{every node}=[draw, shape=circle, minimum size=3pt,inner sep=0pt,
		color=blue, fill=blue]
\node[label=left:$a_1$] at (0,0) (a1) {};
\node[label=left:$a_2a_3$] at (0,2) (a3) {};
\node[label=left:$a_4$] at (0,3) (a4) {};
\node[label=left:$a_5$] at (0,4) (a5) {};
\node[label=right:$b_1$] at (3,0) (b1) {};
\node[label=right:$b_2$] at (3,1) (b2) {};
\node[label=right:$b_3$~~~] at (3,2) (b3) {};
\node[label=right:$b_4$] at (3,3) (b4) {};
\node[label=right:$b_5$] at (3,4) (b5) {};
\draw (a1) -- (a5) (b1) -- (b5);
\draw (a1) -- (b3); \draw[color=red,thick] (b1) -- (a3);
\draw (a3) -- (b5) (b3) -- (a5) (b1) -- (a4);
\end{tikzpicture}
\qquad
\begin{tikzpicture}[scale=0.56]
\tikzstyle{every node}=[draw, shape=circle, minimum size=3pt,inner sep=0pt,
		color=blue, fill=blue]
\node[label=left:$a_1$] at (0,0) (a1) {};
\node[label=left:$a_2a_3$] at (0,2) (a3) {};
\node[label=left:$a_4$] at (0,3) (a4) {};
\node[label=left:$a_5$] at (0,4) (a5) {};
\node[label=right:$b_1$] at (3,0) (b1) {};
\node[label=right:$b_2b_3$] at (3,2) (b3) {};
\node[label=right:$b_4$] at (3,3) (b4) {};
\node[label=right:$b_5$] at (3,4) (b5) {};
\draw (a1) -- (a5) (b1) -- (b5);
\draw[color=red,thick] (a1) -- (b3) (b1) -- (a3);
\draw (a3) -- (b5) (b3) -- (a5) (b1) -- (a4) (a1) -- (b4);
\end{tikzpicture}
\qquad
\begin{tikzpicture}[scale=0.56]
\tikzstyle{every node}=[draw, shape=circle, minimum size=3pt,inner sep=0pt,
		color=blue, fill=blue]
\node[label=left:$a_1$] at (0,0) (a1) {};
\node[label=left:$a_2a_3$] at (0,2) (a3) {};
\node[label=left:$a_4$] at (0,3) (a4) {};
\node[label=left:$a_5$] at (0,4) (a5) {};
\node[label=right:$b_1$] at (3,0) (b1) {};
\node[label=right:$b_2$] at (3,1) (b2) {};
\node[label=right:$b_3b_4$] at (3,2) (b3) {};
\node[label=right:$b_5$] at (3,4) (b5) {};
\draw (a1) -- (a5) (b1) -- (b5);
\draw (a1) -- (b3); \draw[color=red,thick] (b1) -- (a3);
\draw (a3) -- (b5) (b1) -- (a4) (b2) -- (a5); 
\draw[color=red,thick] (b3) -- (a5);
\end{tikzpicture}
$$
\caption{[top] A poset of width $2$ which has natural twin-width (at least)~$2$.
[bottom] All possible contractions of this poset, up to symmetry, which
have red degree $1$.}
\label{fig:exampletww2}
\end{figure}
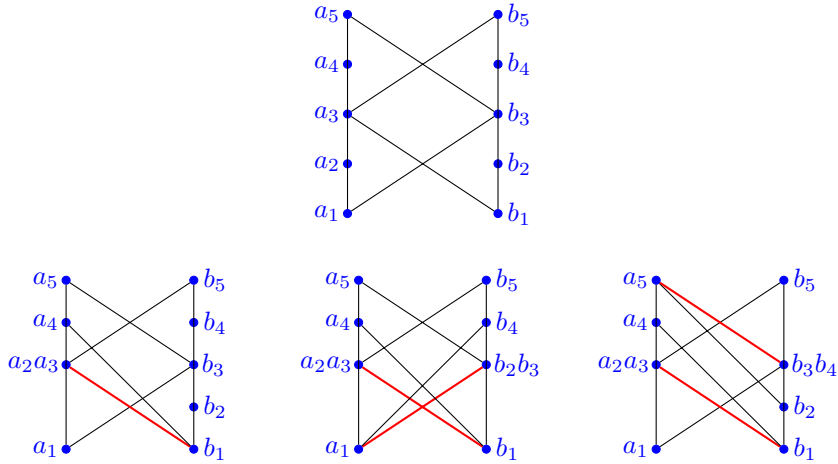

\begin{proposition}
The poset depicted in Figure~\ref{fig:exampletww2}
has the natural twin-width at least~$2$.
\end{proposition}

\begin{proof}
Thanks to symmetries in the depicted poset $P$, it is routine and easy to verify that
every contraction of a pair in $P$ results in two red edges incident
to the contracted vertex, except the contraction of the pair $a_2,\,a_3$
(or the symmetric pairs $b_2,b_3$ or $b_3,b_4$ or $a_3,a_4$).
The result of this contraction, a red poset $P_1$, is on bottom left of Figure~\ref{fig:exampletww2}.

Now, in $P_1$, the contraction of (now red) pair $a_2a_3,\,b_1$ creates
three incident red edges.
For every other contracted pair in $P_1$, we either get the same two red
edges as if it was contracted in $P$, or another red edge incident to
$a_2a_3$ or to $b_1$, or one of the further two possibilities of isolated
red edges depicted at the bottom of Figure~\ref{fig:exampletww2}.
In those cases, again by a boring but routinely easy case analysis,%
\footnote{We have independently verified this conclusion also by an
exhaustive computer check.}
one can see that every contraction creates red degree at least two.
\end{proof}

\section{Conclusions}
We have proved an asymptotically tight relation between the width of a poset
and its twin-width.
The constants in this linear relation are very reasonable, but they can likely
still be improved.
However, the most interesting question for future research is whether the
relation could possible be reversed. This of course cannot be done
directly since there are trivial examples of posets of large width and small
twin-width, but it is an intriguing question of whether classes of bounded
twin-width can always be ``encoded'' (formally, by means of an FO
transduction) in posets of bounded width.
Though, to answer this question one would likely need much stronger tools.

\bibliography{tww}

\end{document}